\documentclass{eptcs}

\providecommand{\event}{SOS 2010}

\usepackage{amssymb}
\usepackage{proof}
\usepackage{stmaryrd}
\usepackage{url}

\newcommand{\Output}[1]{\mathsf{output}~#1}
\newcommand{\Input}[1]{\mathsf{input}~#1}
\newcommand{\Skip}{\ensuremath{\mathsf{skip}}}
\newcommand{\Assign}[2]{\ensuremath{#1 := #2}}
\newcommand{\Seq}[2]{\ensuremath{#1;#2}}
\newcommand{\Ifthenelse}[3]
{\ensuremath{\mathsf{if~} #1 \mathsf{~then~} #2 \mathsf{~else~} #3}}
\newcommand{\While}[2]{\ensuremath{\mathsf{while~} #1 \mathsf{~do~} #2}}
\newcommand{\true}{\ensuremath{\mathsf{true}}}

\newcommand{\eval}[2]{\llbracket #1 \rrbracket #2}
\newcommand{\istrue}[2]{\ensuremath{#2 \models #1}}
\newcommand{\isfalse}[2]{\ensuremath{#2 \not\models #1}}
\newcommand{\update}[3]{\ensuremath{#1[#2\mapsto #3]}}

\newcommand{\res}{\mathit{res}}
\newcommand{\resd}{\mathit{res}_\mathrm{r}}
\newcommand{\resc}{\mathit{res}_\mathrm{c}}

\newcommand{\ret}[1]{\mathit{ret}~#1}
\newcommand{\delay}[1]{\delta~#1}
\newcommand{\inp}[1]{\mathit{in}~#1}
\newcommand{\outp}[2]{\mathit{out}~#1~#2}

\newcommand{\retd}[1]{\mathit{ret}_\mathrm{r}~#1}
\newcommand{\inpd}[1]{\mathit{in}_\mathrm{r}~#1}
\newcommand{\outpd}[2]{\mathit{out}_\mathrm{r}~#1~#2}

\newcommand{\retc}[1]{\mathit{ret}_\mathrm{c}~#1}
\newcommand{\inpc}[1]{\mathit{in}_\mathrm{c}~#1}
\newcommand{\outpc}[2]{\mathit{out}_\mathrm{c}~#1~#2}
\newcommand{\bul}{\bullet}

\newcommand{\up}{\mathit{up}}

\newcommand{\upD}{\mathit{up_r}}
\newcommand{\echoT}{\mathit{echo}}
\newcommand{\echoD}{\mathit{echo'}}
\newcommand{\add}{\mathit{rep}}
\newcommand{\addopt}{\mathit{rep'}}
\newcommand{\addD}{\mathit{rep_r}}

\newcommand{\summ}{\mathop{sum}}
\newcommand{\SmultL}{\mathsf{mult}}
\newcommand{\SmultH}{\mathsf{mult\_opt}}
\newcommand{\mult}{\mathit{mult}}

\newcommand{\Scount}{\mathsf{count}}
\newcommand{\R}[3]{(#1,#2)~R~#3}
\newcommand{\Rnoargs}{R}

\newcommand{\Res}{\mathit{lconf}}
\newcommand{\Ret}[1]{\mathit{\underline{ret}}~#1}
\newcommand{\Delay}[2]{\underline{\delta}~#1~#2}
\newcommand{\Inp}[2]{\mathit{\underline{in}}~#1~#2}
\newcommand{\Outp}[3]{\mathit{\underline{out}}~#1~#2~#3}

\newcommand{\bism}[2]{#1 \approx #2}
\newcommand{\wkbism}[2]{#1 \cong #2}

\newcommand{\Wkbism}[2]{#1 \cong #2}
\newcommand{\steps}[2]{#1 \downarrow #2}    
\newcommand{\uresp}[1]{#1\, {\uparrow}}     
\newcommand{\resp}[1]{#1\, {\Downarrow}}    
\newcommand{\tresp}[1]{#1\, {\Updownarrow}} 

\newcommand{\wkbismi}[3]{#2 \mathbin{{\downarrow}#1{\downarrow}} #3}
\newcommand{\wkbismc}[2]{#1 \mathbin{\cong^\circ} #2}
\newcommand{\cwkbism}[2]{#1 \cong_\mathrm{c} #2}

\newcommand{\wkbismnoargs}{\cong }
\newcommand{\wkbismcnoargs}{{\cong^\circ}}
\newcommand{\wkbisminoargs}[1]{{\downarrow}#1{\downarrow}}
\newcommand{\cwkbismnoargs}{\cong_\mathrm{c}}

\newcommand{\exec}[3]{(#1,#2) \Rightarrow #3}
\newcommand{\execD}[3]{(#1,#2) \mathbin{\Rightarrow_\mathrm{r}} #3}
\newcommand{\execC}[3]{(#1,#2) \mathbin{\Rightarrow_\mathrm{c}} #3}

\newcommand{\execDnoargs}{{\Rightarrow_\mathrm{r}}}
\newcommand{\execCnoargs}{{\Rightarrow_\mathrm{c}}}

\newcommand{\execseq}[3]{(#1,#2) \stackrel{*}{\Rightarrow} #3}
\newcommand{\X}[3]{(#1,#2) \mathbin{X} #3}
\newcommand{\execc}[4]{(#2,#3) \mathbin{{\Rightarrow}{\downarrow}(#1)} #4}
\newcommand{\execd}[3]{(#2,#3) ~{\Rightarrow}{\uparrow}(#1)}

\newcommand{\execcni}[3]{(#1,#2) \mathbin{{\Rightarrow}{\downarrow}} #3}
\newcommand{\execdni}[2]{(#1,#2) ~{\Rightarrow}{\uparrow}}

\newcommand{\execcnoargs}{{{\Rightarrow}{\downarrow}}}
\newcommand{\execseqnoargs}{{\stackrel{*}{\Rightarrow}}}
\newcommand{\execdnoargs}{{\Rightarrow}{\uparrow}}

\newcommand{\cnf}[2]{(#1, #2)}

\newcommand{\step}[3]{(#1, #2) \rightarrow #3}
\newcommand{\stepnoargs}{{\rightarrow}}

\newcommand{\redm}[3]{(#1,#2) \rightsquigarrow #3}

\newcommand{\norm}{\mathit{norm}}
\newcommand{\emb}{\mathit{emb}}

\newcommand{\st}{\sigma}
\newcommand{\state}{\mathit{state}}
\newcommand{\Int}{\mathit{Int}}

\newcommand{\rar}{\rightarrow}
\newcommand{\infix}[3]{#2 ~#1 ~#3}

\newcommand{\cmp}[2]{#1 \circ #2}
\newcommand{\maps}[1]{\sigma[#1]}
\newcommand{\map}[2]{#1\mapsto#2}

\newtheorem{lemma}{Lemma}[section]
\newtheorem{corollary}{Corollary}[section]
\newtheorem{proposition}{Proposition}[section]

\newenvironment{proof}{\noindent \textbf{Proof}}{\hfill $\Box$}

\begin{document}

\title{Resumptions, Weak Bisimilarity and Big-Step Semantics for While with
Interactive I/O: \\ An Exercise in Mixed Induction-Coinduction}

\def\titlerunning{Resumptions, Weak Bisimilarity and Big-Step Semantics for While}

\author{
Keiko Nakata and Tarmo Uustalu
\institute{Institute of Cybernetics at Tallinn University of Technology, 
Akadeemia tee 21, EE-12618 Tallinn, Estonia} 
\email{\{keiko|tarmo\}@cs.ioc.ee}
}

\def\authorrunning{K.~Nakata \& T.~Uustalu}

\maketitle

\begin{abstract}
  We look at the operational semantics of languages with interactive
  I/O through the glasses of constructive type theory. Following on
  from our earlier work on coinductive trace-based semantics for While
  \cite{NU:trabco}, we define several big-step semantics for While
  with interactive I/O, based on resumptions and termination-sensitive
  weak bisimilarity. These require nesting inductive definitions in
  coinductive definitions, which is interesting both mathematically
  and from the point-of-view of implementation in a proof assistant.
  
  After first defining a basic semantics of statements in terms of
  resumptions with explicit internal actions (delays), we introduce a
  semantics in terms of delay-free resumptions that essentially
  removes finite sequences of delays on the fly from those resumptions
  that are responsive. Finally, we also look at a semantics in terms
  of delay-free resumptions supplemented with a silent divergence
  option.  This semantics hinges on decisions between convergence and
  divergence and is only equivalent to the basic one
  classically.
  We have fully formalized our development in Coq.
\end{abstract}

\section{Introduction}

\emph{Interactive} programs are those programs that take inputs, do some
computation, output results, and iterate this cycle possibly infinitely. 
Operating systems and data base systems are typical examples. 
They are important programs and have attracted formal study
to guarantee their correctness/safety. For instance, 
a web application should protect confidentiality of the data it  processes 
in interaction with possibly untrusted agents, 
and a certified compiler should preserve 
input/output behavior of the source program in the compiled code.
These works call for formal semantics of interactive programs. 

Continuing our previous work \cite{NU:trabco} on a trace-based
coinductive big-step semantics for potentially nonterminating
programs, we present a \emph{constructive} account of interactive
input-output \emph{resumptions}\footnote{The word `resumption' is sometimes
  reserved for denotations of parallel threads. We apply it
  more liberally to datastructures recording evolution in
  small steps. This usage dates back to Plotkin
  \cite{Plo:dom} and was reinforced by Cenciarelli and Moggi~\cite{CM:synmds}.}, their important properties, such as weak
bisimilarity and responsiveness (a program always eventually
performs input or output unless it terminates) and \emph{big-step semantics}
of reactive programs. We devise both constructive-style and
classical-style concepts and identify their relationships.
Classical-style concepts rely on upfront decisions of whether a
computation is going to terminate, make an observable action, i.e.,
perform input or output, or silently diverge. The problem
is generally undecidable. Hence, classical-style concepts tend to be
too strong for constructive reasoning.

Our operational semantics are \emph{resumption-based}.  A resumption
is roughly a tree representing possible runs of a program.  The tree
branches on inputs, each edge corresponding to each possible input,
and has infinitely deep paths if the program may diverge.  We begin
the paper by formalizing important properties of resumptions, among
which (termination-sensitive) weak bisimilarity is the most
interesting one technically, requiring nesting of induction into
coinduction.  We give a constructive-style formulation of weak
bisimilarity and relate it to the classical-style version adapted from
previous work~\cite{KM:weabfs,BPSWZ:rean}.  We then present three
\emph{big-step semantics} for interactive While, i.e., While extended
with input/output statements: a basic semantics which explicitly deals
with internal actions (delay steps) and assigns a resumption for all
configurations (statement-state pairs); a delay-free semantics for
responsive configurations; and a classical-style semantics, which is
total classically for all configurations. The two latter semantics
collapse finite sequences of delay steps on the fly. The
classical-style semantics can in addition recognize silent divergence;
classical-style resumptions include a distinguished element to
represent divergence.  Moreover, all three semantics are equivalent
under suitable assumptions. Our approach with big-step semantics in
terms of resumptions allow for reasoning about operational behaviors
of programs in a syntax-independent way. We therefore argue that it is
more abstract than approaches by means of small-step semantics, or
labelled transition systems (in terms of configurations involving a
residual program or a control point).  To compare our big-step
semantics to more traditional approaches, we also define an
uncontroversial small-step semantics with an associated notion of weak
bisimilarity of configurations and show that it agrees with our basic
big-step semantics. These technical results form the main
contributions of the paper.

Why do we want to be \emph{constructive}? First, let us state that our
choice is neither motivated nor depends on any argument of truth: we
are not claiming in this paper that classical logic is less true than
intuitionistic logic and none of the points we make hinge on this
being the case.  Nevertheless, we do think that working in a
constructive logic is very useful also if one has no philosophical
problem in accepting non-constructive arguments. Our reasons are
these. For us, using constructive logic is primarily a technical way
to be conscious about the principles we depend on in our arguments. We
are by no means limiting ourselves: when we really need some
non-constructive principle in a constructive argument, we can always
explicitly assume this principle (or the specific instance that we
need).  But it so happens that a need for unexpectedly strong
principles is often a sign of some imperfect design choice in the
setup of a theory.
Another reason to be constructive as a semanticist is
that programming is about computable functions only.  In constructive
logic, we do not have to specifically worry about computability: only
computable functions are there and can be spoken about. For example,
the formula $\forall x.\, p\, x \vee \neg (p\, x)$ is not a tautology,
it states that $p$ is a decidable: there is a computable function
mapping any $x$ to a proof of $p\, x$ or a proof of $\neg (p\, x)$ (so
also to yes or no, should one not care about the proofs). In big-step
semantics, although we specify evaluation as a relation in this paper,
it is important for us that it can be turned into a function, or else
we do not capture the intuitive idea that programs represent
computable functions from initial configurations into behaviors.

We have \emph{formalized} the development in Coq version 8.2pl1.
This gives us greater confidence in the correctness of our reasoning, 
in particular regarding the productivity of coinductive proofs, since 
the type checker
of Coq helps us avoid mistakes by ruling out improductivity. 
We rely on Mendler-style coinduction to 
circumvent the limitations imposed by 
syntactic guardedness
approach~\cite{syntacticguarded} of Coq. 
The Coq development is available at \url{http://cs.ioc.ee/~keiko/sophie.tgz}.

The language we consider is the While language extended with input and
output primitives, with statements $s : \mathit{stmt}$ defined inductively by
\[
\small
s ::= 
\Skip \mid \Seq{s_0}{s_1} \mid \Assign{x}{e} \mid \Ifthenelse{e}{s_t}{s_f}
\mid \While{e}{s_t} \mid \Input{x} \mid \Output{e}
\]
We assume given the sets of variables and (pure) expressions, whose
elements are ranged over by metavariables $x$ and $e$ respectively. We
assume the set of values to be the integers, non-zero integers
counting as truth and zero as falsity.  The metavariable $v$ ranges
over values. 
A state, ranged over by $\st$, maps variables to values.
The notation $\update{\st}{x}{v}$ denotes the update of a state $\st$ with $v$
at $x$.  We assume given an evaluation function $\eval{e}{\st}$, which
evaluates $e$ in the state $\st$. We write $\istrue{e}{\st}$ and
$\isfalse{e}{\st}$ to denote that $e$ is true, resp.\ false in $\st$.


\section{Resumptions}
\label{sec:resumption}

We will define operational semantics of interactive While in terms of
states and (interactive input-output) resumptions. Informally, a
resumption is a datastructure that captures all possible evolutions of
a configuration (a statement-state pair), a computation tree branching
according to the external non-determinism resulting from interactive
input.\footnote{There are alternatives. We could have chosen to work,
  e.g., with functions from streams of input values into traces, i.e.,
  computation paths.}

Basic (delayful) \emph{resumptions} $r : \res$ are defined
coinductively by the rules\footnote{We mark inductive definitions by single horizontal rules and coinductive definitions by
  double horizontal rules.}
\[
\small
\infer={ \ret{\st} : \res}{ \st : \state}
\quad
\infer={
  \inp{f} : \res
}{
  f : \Int \rar\ \res
}
\quad
\infer={
  \outp{v}{r} : \res
}{
  v : \Int & r : \res
}
\quad
\infer={
  \delay{r} : \res
}{
  r : \res
}
\]
so a resumption either has terminated with some final state,
$\ret{\st}$, takes an integer input $v$ and evolves into a new
resumption $f~v$, $\inp{f}$, outputs an integer $v$ and evolves into
$r$, $\outp{v}{r}$, or performs an internal action (observable at best
as a delay) and becomes $r$, $\delay{r}$.  For simplicity, we assume
input totality; i.e., input resumptions, represented by total
functions, accept any integers. But we could instead have had them
partial, e.g., by letting the constructor $\mathit{in}$ take the
intended domain of definedness as an additional argument.  
We also
define \emph{(strong) bisimilarity} of two resumptions,
$\bism{r}{r_*}$, coinductively by
\[
\small
\infer={
  \bism{\ret{\st}}{\ret{\st}}
}{}
\quad
\infer={
  \bism{\inp{f}}{\inp{f_*}}
}{
  \forall v.\, \bism{f~v}{f_*~v}
}
\quad
\infer={
  \bism{\outp{v}{r}}{\outp{v}{r_*}}
}{
  \bism{r}{r_*}
}
\quad
\infer={
  \bism{\delay{r}}{\delay{r_*}}
}{
  \bism{r}{r_*}
}
\]
Bisimilarity is straightforwardly seen to be an equivalence. We think
of bisimilar resumptions as equal, i.e., type-theoretically we treat
resumptions as a setoid with bisimilarity as the equivalence
relation\footnote{Classically, strong bisimilarity is equality. But we
work in an intensional type theory where strong bisimilarity of
colists is weaker than equality (just as equality of two functions on
all arguments is weaker than equality of these two
functions). E.g., $\bot$ and $\delay{\bot}$ are only strongly bisimilar.}.
Accordingly, we have to make sure that all functions and
predicates we define on resumptions are setoid functions and
predicates, i.e., insensitive to bisimilarity.

Here are some examples of resumptions, defined by corecursion:
\[
\begin{array}{rclrcl}
\bot &=& \delay{\bot}\\
\add~n &=& \delay{(\delay{(\outp{n}{(\add~n)})})}\\
\addopt~n &=& \delay{(\outp{n}{(\addopt~n)})}\\
\echoT~\st &=& \inp{(\lambda n.~\delay{(\mathit{if} ~n\neq 0~\mathit{then}~\outp{n}(\echoT~\st) 
~\mathit{else} ~\ret{\st})})}
\\
\echoD &=& \inp{(\lambda n.~\delay{(\mathit{if} ~n\neq 0~\mathit{then}~\outp{n}\echoD
~\mathit{else} ~\bot)})}
\end{array}
\]
$\bot$ represents a resumption that silently diverges.
$\add$ outputs an integer $n$ forever.
$\addopt$ is similar but has shorter latency.
Both $\echoT$ and $\echoD$ echo input interactively;
the former terminates when the input is 0, whereas the latter diverges in this situation. 

\emph{Convergence}, $\steps{r}{r'}$, states that $r$ converges 
in a finite number of steps to a resumption $r'$, 
which has terminated or 
makes an observable action (performs input/output) as its first move. 
It is defined inductively by
\[
\small
\infer{
  \steps{\ret{\st}}{\ret{\st}}
}{}
\quad
\infer{
  \steps{\inp{f}}{\inp{f_*}}
}{
  \forall v.\, \bism{f~v}{f_*~v}
}
\quad
\infer{
  \steps{\outp{v}{r}}{\outp{v}{r_*}}
}{
  \bism{r}{r_*}
}
\quad
\infer{
  \steps{\delay{r}}{r'}
}{
  \steps{r}{r'}
}
\]
In contrast, \emph{divergence}, $\uresp{r}$, states that
$r$ diverges silently. It is defined coinductively by
\[
\small
\infer={
  \uresp{\delay{r}}
}{
  \uresp{r}
}
\]
For instance, we have 
$\steps{\delay{(\delay{(\ret{\st})})}}{\ret{\st}}$,
$\steps{\add~n}{\outp{n}{(\add~n)}}$  and $\uresp{\bot}$.

Both convergence and divergence are setoid predicates. 
Constructively, it
is not the case that $\forall r.\, 
(\exists r'.\, \steps{r}{r'}) \vee \uresp{r}$, 
which amounts to decidability of convergence. But classically, this dichotomy is true.
In particular, 
$\forall r.\,\neg\, (\exists r'.\, \steps{r}{r'})
\to \uresp{r}$
is constructively provable, but 
$\forall r.\, \neg\, \uresp{r} \to \exists r'.\, \steps{r}{r'}$ holds only classically. 

We can now introduce a useful notion of \emph{responsiveness}.  A
resumption $r$ is responsive, if it keeps converging.  It is defined
coinductively with the help of the convergence predicate by
\[
\small
\infer=[\mathsf{[resp\mbox{-}ret]}]{
  \resp{r}
}{
  \steps{r}{\ret{\st}}
}
\quad
\infer=[\mathsf{[resp\mbox{-}in]}]{
  \resp{r}
}{
  \steps{r}{\inp{f}} & \forall v.\, \resp{(f~v)}
}
\quad
\infer=[\mathsf{[resp\mbox{-}out]}]{
  \resp{r}
}{
  \steps{r}{\outp{v}{r'}} & \resp{r'}
}
\]
For instance, $\add~n$, $\addopt~n$ and $\echoT~\st$ are responsive, 
but $\bot$ and $\echoD$ are not. 

Classically, a resumption is responsive, if it can never evolve into a
diverging resumption.  Indeed, by augmenting the definition of
responsiveness with a divergence option we obtain a classically
tautological predicate, $\tresp{r}$, that we call \emph{commitedness}.
\[
\small
\begin{array}{c}
\infer=[\mathsf{[comm\mbox{-}ret]}]{
  \tresp{r}
}{
  \steps{r}{\ret{\st}}
}
\quad
\infer=[\mathsf{[comm\mbox{-}in]}]{
  \tresp{r}
}{
  \steps{r}{\inp{f}} & \forall v.\, \tresp{(f~v)}
}
\quad
\infer=[\mathsf{[comm\mbox{-}out]}]{
  \tresp{r}
}{
  \steps{r}{\outp{v}{r'}} & \tresp{r'}
}
\quad
\infer=[\mathsf{[comm\mbox{-}div]}]{
  \tresp{r}
}{
  \uresp{r}
}
\end{array}
\]
For a resumption $r$ to be committed, it must be the case that it
always either converges or diverges.  So, classically, any resumption
is committed.

\begin{lemma}
\label{lemma:commitedtotal}
Classically, for all $r$, $\tresp{r}$.
\end{lemma}
\begin{proof}
Specifically, we use an instance of excluded middle, 
$\forall r.\, (\exists r'.\, \steps{r}{r'}) 
\vee \neg (\exists r'.\, \steps{r}{r'})$, 
which amounts to assuming that convergence is decidable. 
\end{proof}

\begin{lemma}
Convergence, divergence, responsiveness and committedness
are setoid predicates. 
\end{lemma}


\section{Weak Bisimilarity}

Two resumptions are weakly bisimilar, if they are bisimilar modulo
collapsing finite sequences of delay steps between observable actions.
It is conceivable that, in practice, weak bisimilarity is what is
needed: one may well be interested only in observable behavior,
disregarding finite delays. For instance, to guarantee correctness of
a compiler optimization, we would want to prove that the optimization
does not change the observable behavior of the source program,
including termination and divergence behaviors, but the optimized code
may perform fewer internal steps and thus be faster.  We therefore
formalize \emph{termination-sensitive} weak bisimilarity, which
distinguishes termination and silent divergence.

Technically, getting the definition of weak bisimilarity right is not
straightforward, especially not in a constructive setting. It requires
both induction and coinduction: we need to collapse a \emph{finite}
number of delay steps between observable actions possibly
\emph{infinitely}. Here we present two equivalent formulations (actually,
we will also give a third one for classical reasoning, which is only
equivalent to the first two classically).  The first is closer to the
formulations typically found in process calculi literature (except
that, in process calculi, one usually works with
termination-insensitive weak bisimilarity).  The second nests
induction into coinduction, exhibiting a useful technique for
implementation in Coq. In our development, we use both formulations
and their equivalence result,
freely choosing the one of the two that facilitates the proof.

The first one, noted $\wkbism{r}{r_*}$, uses coinduction atop the
inductive definition of convergence and is defined by the rules
\[
\small
\begin{array}{c}
\infer={
  \wkbism{r}{r_*}
}{
  \steps{r}{\ret{\st}} & \steps{r_*}{\ret{\st}}
}
\quad
\infer={
  \wkbism{r}{r_*}
}{
  \steps{r}{\inp{f}} & \steps{r_*}{\inp{f_*}} &
  \forall v.\, \wkbism{f~v}{f_*~v}
}
\quad 
\infer={
  \wkbism{r}{r_*}
}{
  \steps{r}{\outp{v}{r'}} & \steps{r_*}{\outp{v}{r_*'}} &
  \wkbism{r'}{r_*'}
}
\quad
\infer={
  \wkbism{\delay{r}}{\delay{r_*}}
}{
  \wkbism{r}{r_*}
}
\end{array}
\]
so two resumptions are weakly bisimilar if they converge at the same action or can both perform an internal action, with weakly bisimilar residual resumptions.
In particular, two terminating resumptions are derived to be 
weakly bisimilar by a single application of the first rule,
whereas two silently diverging resumptions are weakly bisimilar by corecursive application of the fourth rule.
For instance, we have $\wkbism{\add~n}{\addopt~n}$
but $\wkbism{\echoT~\st}{\echoD}$ does \emph{not} hold. 

\begin{lemma}
For any $r, r'$ and $r_*$, if $\steps{r}{r'}$ and $\uresp{r_*}$
then $\neg~\wkbism{r}{r_*}$.
\end{lemma}
As a corollary, we obtain that 
the silently diverging resumption $\bot$ and
resumptions that have terminated, $\ret{\st}$, are not weakly bisimilar. 

The second formulation, denoted $\wkbismc{r}{r_*}$ nests induction
into coinduction. We first define $\wkbisminoargs{X}$ inductively in
terms of $X$, for any relation (read: setoid relation) $X$, and then
define $\wkbismcnoargs$ coinductively in terms of
$\wkbisminoargs{\wkbismcnoargs}$.  For binary relations $X$, $Y$, $X
\subseteq Y$ denotes $\forall x, x_*.\, x \mathbin{X} x_* \to x
\mathbin{Y} x_*$.
\[
\small
\begin{array}{c}
\infer{\wkbismi{X}{\ret{\st}}{\ret{\st}}}{}
\quad
\infer{\wkbismi{X}{\outp{v}{r}}{\outp{v}{r_*}}}{
  \infix{X}{r}{r_*}
}
\quad
\infer{\wkbismi{X}{\inp{f}}{\inp{f_*}}}{
  \forall v.\, \infix{X}{f ~v}{f_* ~v}
}
\quad
\infer{\wkbismi{X}{\delay{r}}{r_*}}{
  \wkbismi{X}{r}{r_*}
}
\quad
\infer{\wkbismi{X}{r}{\delay{r_*}}}{
  \wkbismi{X}{r}{r_*}
}
\\[1ex]
\infer={\wkbismc{r}{r_*}
}{
  X \subseteq \wkbismcnoargs  
  & \wkbismi{X}{r}{r_*}
}
\quad
\infer={\wkbismc{\delay{r}}{\delay{r_*}}}{
  \wkbismc{r}{r_*}
}
\end{array}
\]
Intuitively, $\wkbismi{X}{r}{r_*}$ means that $r$ and $r_*$ converge
to resumptions related by $X$. 
 
In the first rule of $\wkbismcnoargs$, we have used Mendler-style
coinduction 
in order to enable Coq's syntactic guarded corecursion for
$\wkbismcnoargs$.
The natural (Park-style) rule to
stipulate would have been:
\[
\small
\infer={\wkbismc{r}{r_*}
}{
  \wkbismi{\wkbismcnoargs}{r}{r_*}
}
\]
Coq's guardedness condition for induction nested into
coinduction is too weak to work with the Park-style rule: we cannot construct the corecursive functions (coinductive proofs) that we need.
With our definition, the Park-style rule is derivable. 
We can also prove that
$\wkbisminoargs{X}$ is monotone in $X$, which allows us to recover the
natural inversion principle for $\wkbismcnoargs$.

Induction and coinduction can be mixed in several ways. An inductive
definition can be \emph{mutual} with a coinductive definition, if the
occurrence of one predicate in the definition of the other is
contravariant.\footnote{This means looking for a least $X$ and
  greatest $Y$ solving a system of equations $X = F(Y, X)$, $Y = G(X,
  Y)$, where $F$ and $G$ are contravariant in their first arguments
  and covariant in the second arguments.} But this is not our
situation. Instead, in our case, we have an inductive and a
coinductive definition that use each other covariantly, but one is
\emph{nested} in the other. Specifically, we have the inductive
definition nested in the coinductive definition\footnote{i.e., we have
  a definition of the form $\nu X.\, G(\mu Y.\, F(X, Y), X)$ with both $F$
  and $G$ covariant in both arguments}, since we want finite
chunks of $\wkbisminoargs{\wkbismcnoargs}$ derivations to be weaved
into an infinite $\wkbismcnoargs$ derivation. The Agda developer
community is currently exploring a novel approach to coinductive types
(based on suspension types) \cite{DA:mixic,DA:subdem} where this form of mixing
induction and coinduction is easily encoded while nesting the other
way is problematic.

The two definitions of weak bisimilarity are equivalent. 

\begin{lemma}
For any $r$ and $r_*$, $\wkbism{r}{r_*}$ iff $\wkbismc{r}{r_*}$. 
\end{lemma}

Weak bisimilarity is a setoid predicate and an equivalence relation.

\begin{lemma}
  Weak bisimilarity is a setoid predicate: For any $r$, $r'$, $r_*$,
  $r_*'$, if $\bism{r}{r'}$, $\wkbism{r}{r_*}$ and $\bism{r_*}{r_*'}$,
  then $\wkbism{r'}{r_*'}$. Weak bisimilarity is an equivalence.
\end{lemma}

\begin{proof}
Reflexivity and symmetry are straightforward to prove by coinduction.
Below we sketch the proof for transitivity with the second formulation,
$\wkbismc{r}{r_*}$, to show Mendler-style coinduction working in our favour.
For binary relations $X, Y$, let
$\cmp{X}{Y}$ denote their composition; namely, 
$x \mathbin{(\cmp{X}{Y})} x'$ if there is $x''$ such that
$x \mathbin{X} x''$ and $x'' \mathbin{Y} x'$.
We first prove, by {\em induction}, the transitivity for 
$\wkbisminoargs{X}$, i.e., that,
for any resumptions $r_0, r_1, r_2$ and setoid 
relations $X, Y$, 
if $\wkbismi{X}{r_0}{r_1}$
and $\wkbismi{Y}{r_1}{r_2}$,
then $\wkbismi{(\cmp{X}{Y})}{r_0}{r_2}$.
The transitivity of $\wkbismcnoargs$ states that, 
for any resumptions $r_0, r_1$ and $r_2$,
if $\wkbismc{r_0}{r_1}$
and $\wkbismc{r_1}{r_2}$, then $\wkbismc{r_0}{r_2}$.
The proof of this is by \emph{coinduction} and 
inversion on  $\wkbismc{r_0}{r_1}$
and $\wkbismc{r_1}{r_2}$. We show the main case.
Suppose we have $\wkbismc{r_0}{r_1}$
and $\wkbismc{r_1}{r_2}$, because
$\wkbismi{X}{r_0}{r_1}$ and $\wkbismi{Y}{r_1}{r_2}$
for some $X$ and $Y$ such that $X \subseteq \wkbismcnoargs$
and $Y \subseteq \wkbismcnoargs$.
By the transitivity of $\wkbisminoargs{X}$ (which was proved by induction 
separately above), we obtain $\wkbismi{\cmp{X}{Y}}{r_0}{r_2}$.
Using the coinduction hypothesis, we have $\cmp{X}{Y} \subseteq \cmp{\wkbismcnoargs}{\wkbismcnoargs} \subseteq \wkbismcnoargs$,
which closes the case. 
Notably, the invocation of the coinduction hypothesis here is properly guarded thanks to our use of Mendler's trick.
\end{proof}

As one should expect, strongly bisimilar resumptions are weakly
bisimilar.
\begin{corollary}
For any $r$, $r_*$, $\bism{r}{r_*}$, then $\wkbism{r}{r_*}$.
\end{corollary}
\begin{proof}
  Immediate from $\wkbismnoargs$ being a reflexive setoid predicate.
\end{proof}

Termination-sensitive bisimilarity has previously been considered by
Ku\v{c}era and Mayr \cite{KM:weabfs} and Bohannon et al.~\cite{BPSWZ:rean} 
(but see also Bergstra et al.~\cite{BKO:tswkbism}). 
Their version is best suited for classical
reasoning in the sense that terminating and silently diverging
resumptions are distinguished by an upfront choice between convergence
and divergence. This version of weak bisimilarity, denoted
$\cwkbism{r}{r_*}$, is defined coinductively by
\[
\small
\begin{array}{c}
\infer={\cwkbism{r}{r_*}}{
  \steps{r}{\ret{\st}}
  &\steps{r_*}{\ret{\st}}
}
\quad
\infer={\cwkbism{r}{r_*}}{
  \steps{r}{\outp{v}{r'}}
  &\steps{r_*}{\outp{v}{r'_*}}
  &\cwkbism{r'}{r'_*}
}
\quad 
\infer={\cwkbism{r}{r_*}}{
  \steps{r}{\inp{f}}
  &\steps{r_*}{\inp{f_*}}
  &\forall v.\, \cwkbism{f~v}{f_*~v}
}
\quad
\infer={\cwkbism{r}{r_*}}{
  \uresp{r}
  &\uresp{r_*}
}
\end{array}
\]
Only the fourth rule is different from the rules of $\wkbismnoargs$
and refers directly to divergence.

The classical-style version of weak bisimilarity, $\cwkbismnoargs$, is
stronger than the constructive-style version, $\wkbismnoargs$. The
converse is only true classically.

\begin{lemma}
For any $r$ and $r_*$, if $\cwkbism{r}{r_*}$, then  $\wkbism{r}{r_*}$.
Classically, for any $r$ and $r_*$, if $\wkbism{r}{r_*}$, then  $\cwkbism{r}{r_*}$.
\end{lemma}

We insist on the use constructive-style weak bisimilarity,
$\wkbismnoargs$, in the constructive setting, because the
classical-style notion fails to enjoy some fundamental properties
constructively.

\begin{lemma}
  Classical-style weak bisimilarity is a setoid predicate.
  Classically, it is also an equivalence weaker than strong
  bisimilarity.
\end{lemma}
\begin{proof}
We only prove that $\cwkbismnoargs$ is an equivalence. 
Reflexivity: We prove that for any $r$, $\cwkbism{r}{r}$ by 
coinduction. Classically, we have $\forall r_0.\, (\exists r_0'.\,
\steps{r_0}{r_0'}) \vee \uresp{r_0}$.  
Should $\uresp{r}$ hold, we immediately conclude $\cwkbism{r}{r}$.
Suppose there exists $r'$ such that $\steps{r}{r'}$. 
Moreover suppose $r' = \inp{f}$ for some $f$.  
The coinduction hypothesis 
gives us that for any $v$, $\cwkbism{f~v}{f~v}$, from which
$\cwkbism{r}{r}$ follows. The other cases, i.e., when
$r' = \outp{v}{r''}$ for some $v$ and $r''$ or 
$r' = \ret{\st}$ for some $\st$, are similar.
Symmetry:
We prove constructively that for any $r$ and $r'$, if $\cwkbism{r}{r'}$ 
then $\cwkbism{r'}{r}$ by coinduction and inversion on 
$\cwkbism{r}{r'}$. 
Transitivity: 
We prove constructively that for any $r, r'$ and $r''$, if $\cwkbism{r}{r'}$ 
and $\cwkbism{r'}{r''}$ 
then $\cwkbism{r}{r''}$ by coinduction and inversion on 
$\cwkbism{r}{r'}$ and $\cwkbism{r'}{r''}$.
\end{proof}

Constructively, it is not possible to show classical-style weak
bisimilarity reflexive 
and hence we cannot show any two strong
bisimilar resumptions classical-style weakly bisimilar.

A simple example of a resumption $r$ not classical-style weakly
bisimilar to itself constructively is given by any search process that
is classically total, but cannot be proved terminating constructively,
since no bound on the search can be given. By definition, a resumption
can only be classical-style weakly bisimilar to another if it
terminates or diverges. Constructively, the resumption $r$ is only
nondiverging, we cannot show it terminating.


\section{Big-Step Semantics}
\label{sec:exec}

We now proceed to a first, basic (delayful) big-step operational
semantics for our reactive While in terms of delayful resumptions.
Evaluation $\exec{s}{\st}{r}$, expressing that running a statement $s$
from a state $\st$ produces a resumption $r$, is defined coinductively
by the rules in Figure~\ref{fig:exec}.  The rules for sequence and
while implement the necessary sequencing with the help of extended
evaluation $\execseq{s}{r}{r'}$, expressing that running a statement
$s$ from the last state (if it exists) of an already accumulated
resumption $r$ results in a total resumption $r'$.  Extended
evaluation is also defined coinductively, as the coinductive prefix
closure of evaluation.

Input and output statements produce corresponding resumptions that
perform input or output actions and terminate thereafter.  We consider
assignments and testing of guards of if- and while-statements to
constitute internal actions, observable as delays.  This way we avoid
introducing semantic anomalies, by making sure that any while-loop
always progresses.  But this choice also ensures that evaluation is
total---as we should expect. Given that it is deterministic as
well\footnote{Note that the external nondeterminism resulting from
  input actions is encapsulated in resumptions.}, we can equivalently
turn our relational big-step semantics into a functional one: the
unique resumption for a given configuration (statement-state pair) is 
definable by corecursion.\footnote{This aspect makes our big-step operational 
semantics very close in spirit to denotational semantics, specifically, 
denotational semantics in terms of Kleisli categories, here, the Kleisli 
category of a resumptions monad.}  
This semantics is a straightforward adaptation of the
trace-based coinductive big-step semantics of non-interactive While
from our previous work~\cite{NU:trabco}, where the details can be
found and where we motivate all our design choices (e.g., why $\Skip$
takes no time whereas the boolean guards do; we argue that our design 
is canonical).

\begin{figure}[t]
\[
\small
\begin{array}{c}
\quad
\infer={\exec{\Assign{x}{e}}{\st}{\delay{(\ret{\update{\st}{x}{\eval{e}{\st}}})}}
}{}
\qquad
\infer={\exec{\Skip}{\st}{\ret{\st}}
}{}
\qquad
\infer={
  \exec{\Seq{s_0}{s_1}}{\st}{r'}
}{
  \exec{s_0}{\st}{r}
  &\execseq{s_1}{r}{r'}
}
\\[1ex]
\infer={
  \exec{\Ifthenelse{e}{s_t}{s_f}}{\st}{r}
}{
  \istrue{\st}{e}
  &\execseq{s_t}{\delay{(\ret{\st})}}{r}
}
\quad
\infer={
  \exec{\Ifthenelse{e}{s_t}{s_f}}{\st}{r}
}{
  \isfalse{\st}{e}
  &\execseq{s_f}{\delay{(\ret{\st})}}{r}
}
\\[1ex]
\infer={
  \exec{\While{e}{s_t}}{\st}{r'}
}{
  \istrue{\st}{e}
  &\execseq{s_t}{\delta{(\ret{\st})}}{r}
  &\execseq{\While{e}{s_t}}{r}{r'}
}
\quad
\infer={
  \exec{\While{e}{s_t}}{\st}{\delay{(\ret{\st})}}
}{
  \isfalse{\st}{e}
}
\\[1ex]
\infer={
  \exec{\Input{x}}{\st}{\inp{(\lambda v.\ret{\update{\st}{x}{v}})}}
}{}
\qquad
\infer={
  \exec{\Output{e}}{\st}{\outp{(\eval{e}{\st})}{(\ret{\st})}}
}{}
\\[3ex]
\infer={
  \execseq{s}{\ret{\st}}{r}
}{
  \exec{s}{\st}{r}
}
\qquad
\infer={
  \execseq{s}{\inp{f}}{\inp{f'}}
}{
  \forall v.\, \execseq{s}{f~v}{f'~v}
}
\qquad
\infer={
  \execseq{s}{\outp{v}{r}}{\outp{v}{r'}}
}{
  \execseq{s}{r}{r'}
}
\qquad
\infer={
  \execseq{s}{\delay{r}}{\delay{r'}}
}{
  \execseq{s}{r}{r'}
}
\end{array}
\]
\caption{Basic (delayful) big-step semantics
}
\label{fig:exec}
\end{figure}

\begin{lemma}
  Evaluation is a setoid predicate. It is total and deterministic up
  to bisimilarity.
\end{lemma}

Let us look at some examples. 
We have $\exec{\While{\true}{\Skip}}{\st}{\bot}$ for any $\st$.
I.e., $\While{\true}{\Skip}$ silently diverges. 
We also have 
$\exec{\Input{x};\While{\true}{(\Output{x};\Assign{x}{x+1})}}{\st}
{\inp{(\lambda n.~ \up~n)}}$
where $\up$ is defined corecursively by 
$\up~n = \delay{(\outp{n}{(\delay{(\up~(n+1))})})}$.
I.e., the statement counts up from the given input $n$.
The two delays around every output action 
account for the internal actions of 
the assignment and testing of the boolean guard.
An interactive adder takes two inputs 
and outputs their sum, and repeats this process, that is,
we have 
$\exec{\While{\true}{(\Input{x};\Input{y};\Output{(x+y)})}}{\st}{\summ}$
where $\summ$ is defined cocursively by
$\summ = \delay{(\inp{(\lambda m.~\inp{(\lambda n.~\outp{(m+n)}{\summ})})})}$.

Weak bisimilarity is useful for reasoning about soundness of program 
transformations, where we accept that 
transformations may change the timing of a resumption.
For instance,
we have $\exec{\While{\true}{(\Assign{z}{x};\Output{z})}}{\st}
{\add~(\st~x)}$, where $\add$ is defined corecursively by 
$\add~n = \delay{\delay{(\outp{n}{(\add~n)})}}$, 
and
$\exec{\Assign{z}{x};\While{\true}{\Output{z}}}{\st}
{\delay{(\addopt~(\st~x))}}$, 
where $\addopt$ is defined corecursively by 
$\addopt~n = \delay{(\outp{n}{(\addopt~n)})}$,
with $\wkbism{\add~n}{\addopt~n}$.
The latter resumption is faster than the former, but 
they are weakly bisimilar.
In fact, we can prove $\While{e}{(\Assign{z}{x};s)}$
and $\Assign{z}{x};\While{e}{s}$ to be weakly bisimilar whenever $e$ 
is true of the initial state and $s$ does not change $x$. 
Here, the latter statement is obtained from the former by 
loop-invariant code motion, a well-known compiler optimization;
the optimization preserves the observable behaviour of the source statement,
irrespective of its termination behaviour, which it must respect as well.
We note that
$\Output{1}$ is not observationally 
equivalent to $(\While{\true}{\Skip});\Output{1}$. 
More importantly, $\Output{1}$ is not observationally equivalent to 
$\Output{1};(\While{\true}{\Skip})$, since our weak bisimilarity
is termination-sensitive. 
Of course, we can deal with more interesting program equivalences,
such as the equivalence of
$\SmultL  = \While{\true}{(\Input{x};\Input{y};\Assign{z}{0};
\While{x\neq 0}{(\Assign{z}{z+y};\Assign{x}{x-1})};\Output{z})}$
and 
$\SmultH = \While{\true}{(\Input{x};\Input{y};\Ifthenelse{x\geq 0}{\Output{x*y}}
{(\While{\true}{\Skip})})}$,
slow and fast interactive multipliers, which silently
diverge when given a negative first operand.


\section{Small-Step Semantics}

In this section, we introduce an equivalent small-step semantics and
define weak bisimilarity of configurations (statement-state pairs) in
terms of it.  We then prove two configurations to be weakly bisimilar
if and only if their evaluations produce weakly bisimilar resumptions.

A configuration $\cnf{s}{\st}$ is a pair of a statement and state.
\emph{Labelled configurations} $c : \Res$ are defined by the
sum\footnote{The definition is non-recursive, but we pretend that it
  is inductive, as we also do in Coq.}:
\[
\small
\begin{array}{c}
\infer{\Ret{\st} : \Res}{
  s : \mathit{state}}
\quad
\infer{
  \Inp{s}{g}  : \Res
}{
  s : \mathit{stmt}
  & 
  g : \Int \rar\ \mathit{state}
}
\quad
\infer{
  \Outp{v}{s}{\st} : \Res
}{
  v : \Int 
  & 
  s : \mathit{stmt}
  & 
  \st : \mathit{state}
}
\quad 
\infer{
  \Delay{s}{\st} : \Res
}{
  s : \mathit{stmt}
  & \st : \mathit{state}
}
\end{array}
\]

A terminality predicate/one-step reduction relation $\stepnoargs$ is
defined in Figure~\ref{fig:smallstep} (top half). If $c= \Ret{\st}$,
then the proposition $\step{s}{\st}{c}$ means that the configuration
$(s, \st)$ has terminated at state $\st$. In other cases, it
corresponds to a labelled transition: if $c= \Inp{s'}{g}$, we take an
input $v$ and evolve to a configuration $(s', g~v)$; if $c =
\Outp{v}{s'}{\st'}$, we output $v$ and evolve to $\cnf{s'}{\st'}$; if
$c = \Delay{s'}{\st'}$, the configuration $\cnf{s}{\st}$ evolves to a
configuration $\cnf{s'}{\st'}$ in a delay step.  We have chosen to
label configurations rather than transitions so that labelled
configurations become ``trunks'' of resumptions.

\begin{figure}[t]
\[
\small
\begin{array}{c}
\infer{
  \step{\Assign{x}{e}}{\st}
       {\Delay{\Skip}{(\update{\st}{x}{\eval{e}{\st}})}}
}{}
\qquad 
\infer{\step{\Skip}{\st}{\Ret{\st}}}{}
\\[1ex]
\infer{
  \step{\Seq{s_0}{s_1}}{\st}{c}
}{
  \step{s_0}{\st}{\Ret{\st'}}
  &\step{s_1}{\st'}{c}
}
\quad
\infer{
  \step{\Seq{s_0}{s_1}}{\st}{\Inp{(\Seq{s'_0}{s_1})}{f}}
}{
  \step{s_0}{\st}{\Inp{s'_0}{f}}
}
\quad
\infer{
  \step{\Seq{s_0}{s_1}}{\st}
       {\Outp{v}{(\Seq{s'_0}{s_1})}{\st'}}
}{
  \step{s_0}{\st}{\Outp{v}{s'_0}{\st'}}
}
\quad
\infer{
  \step{\Seq{s_0}{s_1}}{\st}{\Delay{(\Seq{s'_0}{s_1})}{\st'}}
}{
  \step{s_0}{\st}{\Delay{s'_0}{\st'}}
}
\\[1ex]
\infer{
  \step{\Ifthenelse{e}{s_t}{s_f}}{\st}
       {\Delay{s_t}{\st}}
}{
  \istrue{e}{\st}
}
\quad
\infer{
  \step{\Ifthenelse{e}{s_t}{s_f}}{\st}
       {\Delay{s_f}{\st}}
}{
  \isfalse{e}{\st}
}
\\[1ex]
\infer{
  \step{\While{e}{s_t}}{\st}
       {\Delay{(\Seq{s_t}{\While{e}{s_t}})}{\st}}
}{
  \istrue{e}{\st}
}
\quad
\infer{
  \step{\While{e}{s_t}}{\st}
       {\Delay{\Skip}{\st}}
}{
  \isfalse{e}{\st}
}
\\[1ex]
\infer{
  \step{\Input{x}}{\st}{\Inp{\Skip}{(\lambda v.\update{\st}{x}{v})}}}{
}
\qquad
\infer{
  \step{\Output{e}}{\st}{\Outp{(\eval{e}{\st})}{\Skip}{\st}}}{
}
\\[3ex]
\infer={
  \redm{s}{\st}{\ret{\st'}}
}{
  \step{s}{\st}{\Ret{\st'}}
}
\qquad
\infer={
  \redm{s}{\st}{\delay{r}}
}{
  \step{s}{\st}{\Delay{s'}{\st'}}
  &\redm{s'}{\st'}{r}
}
\\[1ex]
\infer={
  \redm{s}{\st}{\inp{f}}
}{
  \step{s}{\st}{\Inp{s'}{g}}
  &\forall v.\, \redm{s'}{g~v}{f~v}
}
\qquad
\infer={
  \redm{s}{\st}{\outp{v}{r}}
}{
  \step{s}{\st}{\Outp{v}{s'}{\st'}}
  &\redm{s'}{\st'}{r}
}
\end{array}
\]
\caption{Small-step semantics}
\label{fig:smallstep}
\end{figure}

Weak bisimilarity of two configurations is defined 
in terms of terminality/one-step reduction. 
Again, convergence, $\steps{\cnf{s}{\st}}{c}$, states that
either $\cnf{s}{\st}$ terminates or performs an observable action 
in a finite number of steps.
It is defined inductively by
\[
\small
\begin{array}{c}
\infer{
  \steps{\cnf{s}{\st}}{\Ret{\st'}}
}{
  \step{s}{\st}{\Ret{\st'}}
}
\quad
\infer{
  \steps{\cnf{s}{\st}}{\Outp{v}{s'}{\st'}}
}{
  \step{s}{\st}{\Outp{v}{s'}{\st'}}
}
\quad
\infer{
  \steps{\cnf{s}{\st}}{\Inp{s'}{g}}
}{
  \step{s}{\st}{\Inp{s'}{g}}
}
\quad
\infer{
  \steps{\cnf{s}{\st}}{c}
}{
  \step{s}{\st}{\Delay{s'}{\st'}}
  &\steps{\cnf{s'}{\st'}}{c}
}
\end{array}
\]
(We overload the same notations for resumptions and configurations
without ambiguity.) Weak bisimilarity on 
configurations is defined coinductively by
\[
\small
\begin{array}{c}
\infer={
  \wkbism{\cnf{s}{\st}}{\cnf{s_*}{\st_*}}
}{
  \steps{\cnf{s}{\st}}{\Ret{\st'}}
  &\steps{\cnf{s_*}{\st_*}}{\Ret{\st'}}
}
\\[1ex]
\infer={
  \wkbism{\cnf{s}{\st}}{\cnf{s_*}{\st_*}}
}{
  \steps{\cnf{s}{\st}}{\Inp{s'}{g}} 
  &\steps{\cnf{s_*}{\st_*}}{\Inp{s'_*}{g_*}}
  &\forall v.\, \wkbism{\cnf{s'}{g~v}}{\cnf{s'_*}{g_*~v}}
}
\\[1ex]
\infer={
  \wkbism{\cnf{s}{\st}}{\cnf{s_*}{\st_*}}
}{
  \steps{\cnf{s}{\st}}{\Outp{v}{\cnf{s'}{\st'}}}
  &\steps{\cnf{s_*}{\st_*}}{\Outp{v}{\cnf{s'_*}{\st'_*}}}
  &\wkbism{\cnf{s'}{\st'}}{\cnf{s'_*}{\st'_*}}
}
\\[1ex]
\infer={
  \wkbism{\cnf{s}{\st}}{\cnf{s_*}{\st_*}}
}{
  \step{s}{\st}{\Delay{s'}{\st'}}
  &\step{s_*}{\st_*}{\Delay{s'_*}{\st'_*}}
  &\wkbism{\cnf{s'}{\st'}}{\cnf{s'_*}{\st'_*}}
}
\end{array}
\]

Two configurations are weakly bisimilar if and only if
their evaluations yield weakly bisimilar resumptions. 

\begin{lemma}
For any $s, s_*, \st$ and $\st_*$, 
$\Wkbism{\cnf{s}{\st}}{\cnf{s_*}{\st_*}}$ iff
there exist $r$ and $r_*$ such that 
$\exec{s}{\st}{r}$ and $\exec{s_*}{\st_*}{r_*}$ and $\wkbism{r}{r_*}$.
\end{lemma}

The evaluation relation of the small-step semantics is defined in 
Figure~\ref{fig:smallstep} (bottom half). 
It is the terminal many-step reduction relation, defined coinductively. 
The proposition 
$\redm{s}{\st}{r}$ means that running $s$ from the state $\st$ 
produces the resumption $r$.

The big-step and small-step semantics are equivalent. 

\begin{proposition}
\label{prop:semequiv}
For any $s$, $\st$ and $r$,
$\exec{s}{\st}{r}$ iff $\redm{s}{\st}{r}$.
\end{proposition}


\section{Delay-Free Big-Step Semantics}

So far we explicitly dealt with delay steps in a fully general and
constructive manner. However, it is also possible to define big-step
semantics in terms of resumptions without delay steps, by collapsing
them on the fly, if they come in finite sequences.  In this section,
we define a delay-free semantics for configurations that lead to
responsive resumptions. 

We define \emph{delay-free resumptions}, $r : \resd$, and their (strong)
bisimilarity coinductively by
\[
\small
\begin{array}{c}
\infer={ \retd{\st} : \resd}{ \st : \state}
\quad
\infer={
  \inpd{f} : \resd
}{
  f : \Int \rar\ \resd
}
\quad
\infer={
  \outpd{v}{r} : \resd
}{
  v : \Int
  &
  r : \resd
}
\\[1ex]
\infer={
  \bism{\retd{\st}}{\retd{\st}}
}{}
\quad
\infer={
  \bism{\inpd{f}}{\inpd{f_*}}
}{
  \forall v.\, \bism{f~v}{f_*~v}
}
\quad
\infer={
  \bism{\outpd{v}{r}}{\outpd{v}{r_*}}
}{
  \bism{r}{r_*}
}
\end{array}
\]

A responsive delayful resumption $r : \res$ can be normalized into a
delay-free resumption by collapsing the finite sequences of delay
steps it has between observable actions.  We define normalization,
$\mathit{norm}: (r:\res) \to \resp{r} \to \resd$, and embedding of
delay-free resumptions into delayful resumptions, $\emb : \resd \to
\res$ by corecursion. In the definition of $\norm$, we examine 
the proof of $\resp{r}$, i.e., $r$'s responsiveness.
\[
\small
\hspace*{-3mm}
\begin{array}{rcl@{\hspace*{0mm}}rcl}
\norm~r~(\mathsf{resp\mbox{-}ret} ~\st ~\_\,) &=& \retd{\st}&
  \emb~(\retd{\st}) & = & \ret{\st} \\
\norm~r~(\mathsf{resp\mbox{-}in} ~f ~\_~k) &=& 
\inpd{(\lambda v.\, \norm~(f~v) ~(k~v))} &
  \emb~(\inpd{f}) & = & \inp{(\lambda v.\ \emb~(f~v))} \\
\norm~r~(\mathsf{resp\mbox{-}out} ~v ~r'~\_~h) 
&=& \outpd{v}{(\norm~r'~h)} &
  \emb~(\outpd{v}{r}) & = & \outp{v}{(\emb~r)}
\end{array}
\]

A delayful resumption is weakly bisimilar to a delay-free one if and
only if it is responsive and its normal form is strongly bisimilar to
the same.

\begin{lemma}
  For any $r : \res$ and $r_* : \resd$, $\wkbism{r}{\emb~r_*}$ iff
  $\bism{\norm~r~h}{r_*}$ for some $h: \resp{r}$.
\end{lemma}
(The convergence proofs of a resumption are strong bisimilar, so $h$
is unique up to that extent.)

\begin{corollary}
  \emph{(i)} For any $r : \res$, $h: \resp{r}$, $\wkbism{r}{\emb~(\norm~r~h)}$.
  \emph{(ii)} For any $r$, $h:\resp{r}$ and $r_*$, $h_*:\resp{r_*}$,
  $\wkbism{r}{r_*}$ iff $\bism{\norm~r~h}{\norm~r_*~h_*}$.
\end{corollary}

In Figure~\ref{fig:nodelayexec}, we define the delay-free big-step
semantics for responsive programs. Here we have an inductive
definition of a parameterized evaluation relation $\execcnoargs(X)$
defined in terms of $X$, for any relation $X$, nested into a
coinductive definition of an extended evaluation relation
$\execseqnoargs$, defined in terms of
${\execcnoargs}(\execseqnoargs)$.  Finally, the actual evaluation
relation $\execDnoargs$ of interest is obtained by instantiating
$\execcnoargs$ at $\execseqnoargs$. Since we collapse delay-steps on
the fly, an assignment immediately terminates at the updated state.
Likewise, testing the guard of a condition or a while-loop takes no
time. The crucial rules are those for sequence and while-loop. If the
first statement of a sequence or the body of a while-loop terminate
silently, the second statement or the new iteration of the loop are
run using the inductive evaluation.  The coinductive extended
evaluation is used only if the first statement or the body perform at
least one input or output action.

This way, we make sure that only a finite number of delay steps may be
collapsed between two observable actions, while allowing for diverging
runs which perform input and output every now and then. Indeed, if we
replaced the while-ret rule with
\[
\small
\infer{
  \execc{X}{\While{e}{s_t}}{\st}{r'}
}{
  \istrue{\st}{e}
  &\execc{X}{s_t}{\st}{\retd{\st'}}
  &\execseq{\While{e}{s_t}}{\retd{\st'}}{r'}
}
\]
we would obtain semantic anomalies. E.g., 
$\execD{\While{\true}{\Skip}}{\st}{r}$ would be derived
for any $r : \resd$.

\begin{figure}[t]
\[
\small
\begin{array}{c}
\quad
\infer{\execc{X}{\Assign{x}{e}}{\st}{\retd{(\update{\st}{x}{\eval{e}{\st}})}}
}{}
\qquad
\infer{\execc{X}{\Skip}{\st}{\retd{\st}}
}{}
\qquad 
\infer{
  \execc{X}{\Seq{s_0}{s_1}}{\st}{r}
}{
  \execc{X}{s_0}{\st}{\retd{\st'}}
  &\execc{X}{s_1}{\st'}{r}
}
\\[1ex]
\infer{
  \execc{X}{\Seq{s_0}{s_1}}{\st}{\inpd{f'}}
}{
  \execc{X}{s_0}{\st}{\inpd{f}}
  &\forall v.\, \X{s_1}{f\, v}{f'\, v}
}
\quad
\infer{
  \execc{X}{\Seq{s_0}{s_1}}{\st}{\outpd{v}{r'}}
}{
  \execc{X}{s_0}{\st}{\outpd{v}{r}}
 &\X{s_1}{r}{r'} 
}
\\[1ex]
\infer{
  \execc{X}{\Ifthenelse{e}{s_t}{s_f}}{\st}{r}
}{
  \istrue{\st}{e}
  &\execc{X}{s_t}{\st}{r}
}
\quad
\infer{
  \execc{X}{\Ifthenelse{e}{s_t}{s_f}}{\st}{r}
}{
  \isfalse{\st}{e}
  &\execc{X}{s_f}{\st}{r}
}
\\[1ex]
\infer{
  \execc{X}{\While{e}{s_t}}{\st}{r}
}{
  \istrue{\st}{e}
  &\execc{X}{s_t}{\st}{\retd{\st'}}
  &\execc{X}{\While{e}{s_t}}{\st'}{r}
}
\\[1ex]
\infer{
  \execc{X}{\While{e}{s_t}}{\st}{\inpd{f'}}
}{
  \istrue{\st}{e}
  &\execc{X}{s_t}{\st}{\inpd{f}}
  &\forall v.\, \X{\While{e}{s_t}}{f\, v}{f'\, v}
}
\quad 
\infer{
  \execc{X}{\While{e}{s_t}}{\st}{\outpd{v}{r'}}
}{
  \istrue{\st}{e}
  &\execc{X}{s_t}{\st}{\outpd{v}{r}}
  &\X{\While{e}{s_t}}{r}{r'}
}
\\[1ex]
\infer{
  \execc{X}{\While{e}{s_t}}{\st}{\retd{\st}}
}{
  \isfalse{\st}{e}
}
\\[1ex]
\infer{
  \execc{X}{\Input{x}}{\st}{\inpd{(\lambda v.\retd{\update{\st}{x}{v}})}}
}{}
\qquad
\infer{
  \execc{X}{\Output{e}}{\st}{\outpd{(\eval{e}{\st})}{(\retd{\st})}}
}{}
\\[3ex]
\infer={
  \execseq{s}{\retd{\st}}{r}
}{
  X \subseteq \execseqnoargs
  & 
  \execc{X}{s}{\st}{r}
}
\qquad
\infer={
  \execseq{s}{\inpd{f}}{\inpd{f'}}
}{
  \forall v.\, \execseq{s}{f~v}{f'~v}
}
\qquad
\infer={
  \execseq{s}{\outpd{v}{r}}{\outpd{v}{r'}}
}{
  \execseq{s}{r}{r'}
}
\\[3ex]
\infer{\execD{s}{\st}{r}}{
  \execc{\execseqnoargs}{s}{\st}{r}
}
\end{array}
\]
\caption{Delay-free big-step semantics
}
\label{fig:nodelayexec}
\end{figure}

Coming back to the examples of the previous section,
we have
$\execD{\Input{x};\While{\true}{(\Output{x};\Assign{x}{x+1})}}{\st}
{\inp{(\lambda n.~\upD~n)}}$
where $\upD$ is defined corecursively by 
$\upD~n = \outpd{n}{(\upD~(n+1))}$.
We also have
$\execD{\While{\true}{\Assign{z}{x};\Output{z}}}{\st}
{\addD~(\st~x)}$
and 
$\exec{\Assign{z}{x};\While{\true}{\Output{z}}}{\st}
{\addD~(\st~x)}$
where $\addD$ is defined corecursively by
$\addD~n = \outpd{n}{(\addD~n)}$.
Since the delay steps are collapsed on the fly in the delay-free semantics,  
the two statements produce the same, i.e., strongly bisimilar, (delay-free) 
resumptions. 
The delay-free semantics does not account for (i.e., does not assign a 
resumption to) non-responsive configurations, such as 
$\While{\true}{\Skip}$ and the interactive multipliers from the
previous section (since they diverge given a negative input for the first
operand), with any initial state.

We state adequacy of the delay-free semantics by relating it to the
delayful semantics of Section~\ref{sec:exec}.  Namely, for
configurations leading to responsive resumptions they agree.

\begin{proposition}[Soundness]
For any $s$, $\st$, $r : \resd$, if $\execD{s}{\st}{r}$
then there exists $r' : \res$ such that $\exec{s}{\st}{r'}$
and $\wkbism{\emb~r}{r'}$.
\end{proposition}

\begin{proposition}[Completeness]
\label{prop:execc_complete}
For any $s$, $\st$, $r : \res$ and $h:\resp{r}$, if $\exec{s}{\st}{r}$, then
$\execD{s}{\st}{\norm~r~h}$.
\end{proposition}

The proofs are omitted due to the
space limitation. They are nontrivial and the details can be found in
the accompanying Coq development. Below we demonstrate 
the key proof technique on an example. 

Consider the statement $\Scount = \While{\true}{(\Ifthenelse{i >
0}{\Assign{i}{i-1}} {(\Output{x};\Assign{x}{x+1};\Assign{i}{x})})}$. It
counts up from 0, so we should have $\execD{\Scount}{\st}{\upD~0}$ for
an initial state $\st$ that maps $x$ and $i$ to 0. 
We need coinduction since $\Scount$ performs outputs infinitely often;
we also need induction, nested into coinduction, since 
the loop silently iterates $n$ times each time before outputting $n$.
Note that the latency is finite but unbounded.

We cannot perform induction inside coinduction na\"ively. That would 
be rejected by Coq's syntactic
guardedness checker, which is there to ensure productivity of coinduction. 
Mendler-style coinduction comes to rescue. 
Let $\R{s}{r}{r'}$ be a relation on pairs $(s,r)$ of a statement and
a resumption and resumptions $r'$, defined inductively by 
\[
\small
\infer{ \R{\Scount}{\retd{\maps{\map{x}{n},\map{i}{n}}}}{\upD~n} }{}
\quad
\infer{
  \R{s}{\outpd{v}{r}}{\outpd{v}{r'}}
}{
  \R{s}{r}{r'} }
\quad
\infer{ \R{s}{\retd{\st}}{r} }{ \execc{\Rnoargs}{s}{\st}{r} }
\]
The key fact is that $\Rnoargs$ is stronger than $\execseqnoargs$
(Lemma~\ref{lemma:R_imp_execseq} below).

We first prove that $\execcnoargs$ is monotone by induction.

\begin{lemma}\label{lemma:execc_mono}
For any $X$, $Y$, $s$, $\st$ and $r$ such that $X \subseteq Y$,
if $\execc{X}{s}{\st}{r}$, then $\execc{Y}{s}{\st}{r}$.
\end{lemma}

The following two lemmata are proved by 
straightforward application of the rules in
Figure~\ref{fig:nodelayexec}.

\begin{lemma}
For any $n$,
$\execc{\Rnoargs}
{\Ifthenelse{i >0}{\Assign{i}{i-1}}{(\Output{x};\Assign{x}{x+1};\Assign{i}{x})}}
{\maps{\map{x}{n},\map{i}{0}}}
{\outpd{n}{(\retd{\maps{\map{x}{n+1},\map{i}{n+1}}})}}$.
\end{lemma}

\begin{lemma}
For any $n$ and $m$,
$\execc{\Rnoargs} 
{\Ifthenelse{i >0}{\Assign{i}{i-1}} {(\Output{x};\Assign{x}{x+1};\Assign{i}{x})}}
{\maps{\map{x}{n},\map{i}{m+1}}}
{\retd{\maps{\map{x}{n},\map{i}{m}}}}$.
\end{lemma}

The next lemma is proved by induction on $m$,
using the two lemmata just proved. 

\begin{lemma}
For any $n$ and $m$, 
$\execc{\Rnoargs}{\Scount}
{\maps{\map{x}{n},\map{i}{m}}}{\outpd{n}{(\upD~(n+1))}}$.
\end{lemma}

\begin{corollary}\label{coro:execc_upD}
For any $n$,
$\execc{\Rnoargs}{\Scount}{\maps{\map{x}{n},\map{i}{n}}}{\upD~n}$.
\end{corollary}

We can now prove that $R$ is stronger than
$\execseqnoargs$ by coinduction and inversion on $\R{s}{r}{r'}$.
Here is the crux: corollary~\ref{coro:execc_upD} together
with the coinduction hypothesis gives $\execseq{\Scount}
{\retd{\maps{\map{x}{n},\map{i}{n}}}}{\upD~n}$, and the use of 
the coinduction hypothesis is properly guarded. 

\begin{lemma}\label{lemma:R_imp_execseq}
For any $s, r$ and $r'$, if
$\R{s}{r}{r'}$ then $\execseq{s}{r}{r'}$
\end{lemma}

The main proposition follows from 
corollary~\ref{coro:execc_upD}, lemma~\ref{lemma:R_imp_execseq} 
and the monotonicity of $\execcnoargs$ 
(lemma~\ref{lemma:execc_mono}).

\begin{proposition}
For any $n$,
$\execD{\Scount}{\maps{\map{x}{n},\map{i}{n}}}
{\upD~n}$.
\end{proposition}


\section{Classical-Style Big-Step Semantics}

In Section~\ref{sec:resumption}, we augmented the definition of
responsiveness with a divergence option to obtain a concept of
committedness, which is a classically tautological predicate.
Similarly, we can obtain a delay-free semantics for committed
configurations from the delay-free semantics for responsive
configurations of the previous section.  To do so, we extend the
definition of delay-free resumptions with a ``black hole''
constructor, $\bul$, representing silent divergence, arriving at
\emph{classical-style resumptions}, and adjust the definition of
(strong) bisimilarity: 
\[
\small
\begin{array}{c}
\infer={ \retc{\st} : \resc}{ \st : \state}
\quad
\infer={
  \inpc{f} : \resc
}{
  f : \Int \rar\ \resc
}
\quad
\infer={
  \outpc{v}{r} : \resc
}{
  r : \resc
}
\quad 
\infer={\bul : \resc}{}
\\[1ex]
\infer={
  \bism{\retc{\st}}{\retc{\st}}
}{}
\quad
\infer={
  \bism{\inpc{f}}{\inpc{f_*}}
}{
  \forall v.\, \bism{f~v}{f_*~v}
}
\quad
\infer={
  \bism{\outpc{v}{r}}{\outpc{v}{r_*}}
}{
  \bism{r}{r_*}
}
\quad
\infer={\bism{\bul}{\bul}}{}
\end{array}
\]

Given a proof $h: \tresp{r}$ of committedness of a delayful resumption
$r : \res$, we can normalize $r$ into a classical resumption by
collapsing the finite delays between observable actions and sending
silent divergence into the black hole.
\[
\small
\hspace*{-5mm}
\begin{array}{rcl@{\hspace*{0mm}}rcl}
\norm~r~(\mathsf{comm\mbox{-}ret}~\st~\_) &=& \retc{\st} &
  \emb~(\retc{\st}) & = & \ret{\st} \\
\norm~r~(\mathsf{comm\mbox{-}in}~f~\_~k) &=& \inpc{(\lambda v.\, \norm ~(f~v) ~(k~v))} & 
  \emb~(\inpc{f}) & = & \inp{(\lambda v.\ \emb~(f~v))} \\ 
\norm~r~(\mathsf{comm\mbox{-}out}~v~r'~\_~h) &=& \outpc{v}{(\norm ~r' ~h)} &
  \emb~(\outpc{v}{r}) & = & \outp{v}{(\emb~r)} \\
\norm~r~(\mathsf{comm\mbox{-}div}~\_) &=& \bul & 
  \emb~\bul & = & \delay{(\emb~\bul)}
\end{array}
\]

Again, a delayful resumption is weakly bisimilar to a classical-style
one if and only if it is committed and its normal form is strongly
bisimilar.

\begin{lemma}
For any $r : \res$ and $r_* : \resc$, 
$\wkbism{r}{\emb~r_*}$ iff $\bism{\norm~r~h}{r_*}$ for some  $h: \tresp{r}$.
\end{lemma}

In Figure~\ref{fig:classicexec}, we define the classical-style
semantics in terms of classical-style resumptions. We have an
inductive parameterized evaluation relation $\execcnoargs(X)$, defined
in terms of $X$, for any relation $X$, for convergent runs; its
inference rules are the same as those in the previous section. But we
also have a coinductive parameterized evaluation $\execdnoargs(X)$,
again defined in terms of $X$, for any relation $X$, for silently
diverging runs, so that $\execd{\execseqnoargs}{s}{\st}$ expresses
that running a statement $s$ from a state $\st$ diverges without
performing input or output.  It uses the inductive evaluation in case
the first statement of a sequence or the first iteration of the body
of a while-loop silently terminates, but the whole sequence or
while-loop silently diverges. Then we define coinductively
an extended evaluation relation $\execseqnoargs$, in terms of these
two evaluation relations, nesting the latter into the former. Finally,
we instantiate both $\execcnoargs$ and $\execdnoargs$ at
$\execseqnoargs$ to obtain the ``real'' evaluation relation
$\execCnoargs$. Note that, to derive an evaluation proposition in this
semantics, one has to decide upfront whether inductive or coinductive
evaluation should be used---a decision that can be made classically,
but not constructively.

\begin{figure}[t]
\[
\small
\begin{array}{c}
\infer{\execc{X}{\Assign{x}{e}}{\st}{\retc{(\update{\st}{x}{\eval{e}{\st}})}}
}{}
\qquad
\infer{\execc{X}{\Skip}{\st}{\retc{\st}}
}{}
\qquad
\infer{
  \execc{X}{\Seq{s_0}{s_1}}{\st}{r}
}{
  \execc{X}{s_0}{\st}{\retc{\st'}}
  &\execc{X}{s_1}{\st'}{r}
}
\\[1ex]
\infer{
  \execc{X}{\Seq{s_0}{s_1}}{\st}{\inpc{f'}}
}{
  \execc{X}{s_0}{\st}{\inpc{f}}
  & \forall v.\, \X{s_1}{f~v}{f'~v}
}
\quad
\infer{
  \execc{X}{\Seq{s_0}{s_1}}{\st}{\outpc{v}{r'}}
}{
  \execc{X}{s_0}{\st}{\outpc{v}{r}}
  &\X{s_1}{r}{r'}
}
\\[1ex]
\infer{
  \execc{X}{\Ifthenelse{e}{s_t}{s_f}}{\st}{r}
}{
  \istrue{\st}{e}
  &\execc{X}{s_t}{\st}{r}
}
\quad
\infer{
  \execc{X}{\Ifthenelse{e}{s_t}{s_f}}{\st}{r}
}{
  \isfalse{\st}{e}
  &\execc{X}{s_f}{\st}{r}
}
\\[1ex]
\infer{
  \execc{X}{\While{e}{s_t}}{\st}{r}
}{
  \istrue{\st}{e}
  &\execc{X}{s_t}{\st}{\retc{\st'}}
  &\execc{X}{\While{e}{s_t}}{\st'}{r}
}
\\[1ex]
\infer{
  \execc{X}{\While{e}{s_t}}{\st}{\inpc{f'}}
}{
  \istrue{\st}{e}
  &\execc{X}{s_t}{\st}{\inpc{f}}
  &\forall v.\, \X{\While{e}{s_t}}{f\, v}{f'\, v}
}
\quad 
\infer{
  \execc{X}{\While{e}{s_t}}{\st}{\outpc{v}{r'}}
}{
  \istrue{\st}{e}
  &\execc{X}{s_t}{\st}{\outpc{v}{r}}
  &\X{\While{e}{s_t}}{r}{r'}
}
\\[1ex]
\infer{
  \execc{X}{\While{e}{s_t}}{\st}{\retc{\st}}
}{
  \isfalse{\st}{e}
}
\\[1ex]
\infer{
  \execc{X}{\Input{x}}{\st}{\inpc{(\lambda v.\retc{\update{\st}{x}{v}})}}
}{}
\qquad
\infer{
  \execc{X}{\Output{e}}{\st}{\outpc{(\eval{e}{\st})}{(\retc{\st})}}
}{}
\\[3ex]
\infer={
  \execd{X}{\Seq{s_0}{s_1}}{\st}
}{
  \execd{X}{s_0}{\st}
}
\quad
\infer={
  \execd{X}{\Seq{s_0}{s_1}}{\st}
}{
  \execc{X}{s_0}{\st}{\retc{\st'}}
  &\execd{X}{s_1}{\st'}
}
\\[1ex]
\infer={
  \execd{X}{\Ifthenelse{e}{s_t}{s_f}}{\st}
}{
  \istrue{\st}{e}
  &\execd{X}{s_t}{\st}
}
\quad
\infer={
  \execd{X}{\Ifthenelse{e}{s_t}{s_f}}{\st}
}{
  \isfalse{\st}{e}
  &\execd{X}{s_f}{\st}
}
\\[1ex]
\infer={
  \execd{X}{\While{e}{s_t}}{\st}
}{
  \istrue{\st}{e}
  &\execd{X}{s_t}{\st}
}
\quad
\infer={
  \execd{X}{\While{e}{s_t}}{\st}
}{
  \istrue{\st}{e}
  &\execc{X}{s_t}{\st}{\retc{\st'}}
  &\execd{X}{\While{e}{s_t}}{\st'}
}
\\[2.5ex]
\infer={
  \execseq{s}{\retc{\st}}{r}
}{
  X \subseteq \execseqnoargs
  &
  \execc{X}{s}{\st}{r}
}
\quad
\infer={
  \execseq{s}{\retc{\st}}{\bul}
}{
  X \subseteq \execseqnoargs
  &
  \execd{X}{s}{\st}
}
\quad 
\infer={
  \execseq{s}{\inpc{f}}{\inpc{f'}}
}{
  \forall v.\, \execseq{s}{f~v}{f'~v}
}
\quad
\infer={
  \execseq{s}{\outpc{v}{r}}{\outpc{v}{r'}}
}{
  \execseq{s}{r}{r'}
}
\quad
\infer={\execseq{s}{\bul}{\bul}}{}
\\[2.5ex]
\infer{\execC{s}{\st}{r}}{
  \execc{\execseqnoargs}{s}{\st}{r}
}
\quad
\infer{\execC{s}{\st}{\bul}}{
  \execd{\execseqnoargs}{s}{\st}
}
\end{array}
\]

\caption{Classical-style big-step semantics}
\label{fig:classicexec}
\end{figure}

The classical-style semantics is adequate wrt.\ the basic semantics
of Section~\ref{sec:exec}.  

\begin{proposition}[Soundness]
For any $s$, $\st$ and $r : \resc$, if $\execC{s}{\st}{r}$, 
then there exists $r' : \res$ such that
$\exec{s}{\st}{r'}$ and $\wkbism{\emb~r}{r'}$.
\end{proposition}

\begin{proposition}[Completeness]
\label{prop:execexecd}
For any $s$, $\st$ and $r : \res$ and $h : \tresp{r}$, if $\exec{s}{\st}{r}$,
then $\execC{s}{\st}{\norm~r~h}$.
\end{proposition}

\begin{corollary}
\label{coro:execexecd}
Classically, for any $s$, $\st$ and $r : \res$, if $\exec{s}{\st}{r}$,
then there exists $r' : \resc$ such that $\execC{s}{\st}{r'}$ and
$\wkbism{r}{\emb~r'}$.
\end{corollary}

The classical-style semantics is more expressive than
responsive semantics, since it offers the option of ``detected'' divergence.
In particular we have $\execC{\While{\true}{\Skip}}{\st}{\bul}$
and our interactive multipliers are assigned a classical-style 
resumption $\mult$ defined corecursively by 
$\mult = \inpc{(\lambda m.~\inpc{(\lambda n.~ \mathit{if}
~m \geq 0 ~\mathit{then} ~\outpc{(m*n)}{\mult} ~\mathit{else} ~\bul)})}$; i.e., we have 
$\execC{\SmultL}{\st}{\mult}$
and 
$\execC{\SmultH}{\st}{\mult}$.


\section{Related Work}

Formalized semantics are an important ingredient in the trusted
computing base of certified compilers. Proof assistants, like Coq, are
a good tool for such formalization projects, as both the object
semantics of interest and its metatheory can be developed in the same
framework.  For introductions, see \cite{Ber:surpls}.

To account for nontermination or silent divergence properly in a
big-step semantics is nontrivial already for languages without
interactive I/O. Leroy and Grall \cite{LG:coibso} introduced two
big-step semantics for lambda-calculus. One is classical in spirit,
with two evaluation relations, inductive and coinductive, for
terminating and diverging runs, and relies on decidability between
termination and divergence.  The other, with a single coinductive
evaluation relation, is essentially suited for constructive reasoning,
but contains a semantic anomaly (a function can continue reducing
after the argument diverges), which results from its ability to
collapse an infinite sequence of internal actions (contraction steps).

In our work \cite{NU:trabco} on While with nontermination, we
developed a trace-based coinductive big-step semantics where traces
were non-empty colists of intermediate states, agreeing with the
very standard coinductive small-step trace-based semantics. This semantics
relied on traces being a monad; a central component in the definition
was an extended evaluation relation, corresponding to the Kleisli
extension of evaluation. Capretta \cite{Cap:genrct} studied
constructive denotational semantics of nontermination as the Kleisli
semantics for the delayed state monad, corresponding to hiding the
intermediate states in the trace monad as internal actions and
quotienting by termination-sensitive weak bisimilarity. Rutten
\cite{Rut:notcwb} carried out a similar project in classical set
theory where the quotient is the state space extended with an extra
element for nontermination.

Operational semantics of interactive programs is most often described
in the small-step style where it amounts to a labelled transition
system. Especially, this is the dominating approach in process
calculi. Big-step semantics is closer to denotational semantics. In
this field, resumption-based descriptions go back to
Plotkin~\cite{Plo:dom}, Gunter et al.\ ~\cite{GMS:semdds} and
Cenciarelli and Moggi~\cite{CM:synmds}. Resumptions are a monad and
resumptions-based denotational semantics is a Kleisli semantics. Our
big-step semantics are directly inspired by this approach, except that
we work in a constructive setting and must take extra care to avoid
the need to invoke classical principles where they are dispensable.

We are not aware of many other works on constructive semantics of
interactive I/O. But similar in its spirit to ours is the work of
Hancock et al.~\cite{HPG:repspu} on stream processors and the stream
functions that these induce by ``eating''. Stream processors are like
our delay-free resumptions, except that the authors emphasize parallel
composition of stream processors (one processor's output becomes
another processor's input) and, for this to be well-defined, a stream
processor must not terminate and may only do a finite number of input
actions consecutively.  Hancock et al.~\cite{GHP:conffc} also
characterize realizable stream functions. In a precursor work, Hancock
and Setzer \cite{HS:intpdt} studied a model of interaction where a
client sends a server commands and expects responses.

Weak bisimilarity tends to be defined termination-insensitively,
identifying termination and divergence. In particular, this is also
the approach of CCS \cite{Mil:comc}.  Termination-sensitive weak
bisimilarity has been considered by 
Bergstra, Klop and Olderog~\cite{BKO:tswkbism}, 
Ku\v{c}era and Mayr
\cite{KM:weabfs} and Bohannon et al.~\cite{BPSWZ:rean}, but only in
what we call the classical-style version, relying on decisions 
between convergence and divergence.
(The weak bisimilarity of Capretta
\cite{Cap:genrct} is termination-sensitive and tailored for constructive
reasoning, but restricted to behaviours without I/O. Weak bisimilarity also motivated the study of Danielsson and Altenkirch~\cite{DA:mixic} on mixed induction-coinduction.)

Mixed inductive-coinductive definitions in the form of induction
nested into coinduction \linebreak ($\nu X.\, \mu Y.\, F\, (X,Y)$ or, more generally, $\nu X.\, G(\mu Y.\, F\, (X,Y), X)$) seem to be
quite fundamental in applications (e.g., the stream processors of
Hancock et al., our delay-free semantics).  Danielsson and
Altenkirch~\cite{DA:mixic,DA:subdem} argue for making this mix the
basic form of inductive-coinductive definitions in the
dependently-typed programming language Agda.  Unfortunately, nestings
the other way around (definitions $\mu X.\, \nu Y.\, F\, (X,Y)$) seem
to become difficult or impossible to code. With our approach,
coinduction nested into induction is handled
symmetrically to induction nested into coinduction~\cite{NU:mixicw}.

Mendler-style (co)recursion originates from Mendler~\cite{Men:indttc}.
It uses that a monotone (co)inductive definition is equivalent to a
positive one, via a syntactic left (right) Kan extension along
identity (instead of $\mu X.\, F\, X$ one works with $\mu X.\, \exists
Y.\, (Y \to X) \to F\, Y $). We exploited this fact to enable Coq's
guarded corecursion for a coinductive definition with a nested
inductive definition, at the price of impredicativity.

We have previously developed and formalized a Hoare logic for the
trace-based semantics of While with nontermination \cite{NU:hoalct}. A
similar enterprise should be possible for resumptions, weak
bisimilarity and While with interactive I/O.


\section{Conclusion}

We have developed a constructive treatment of resumption-based
big-step semantics of While with interactive I/O. We have devised
constructive-style definitions of important concepts on resumptions
such as termination-sensitive weak bisimilarity and responsiveness,
and devised two variations of delay-free big-step semantics for
programs that produce responsive and committed resumptions,
respectively.  Responsiveness is for interactive computation what
termination is for noninteractive computation. And likewise,
committedness compares to a decided domain of definedness.  Indeed, all
three variations of big-step semantics for While with interactive I/O
have counterparts in big-step semantics for
noninteractive While (see Appendix).  Mathematically, we find it
reassuring that observations made for a more simpler noninteractive
While naturally scale to a more involved language with interactive I/O.  
The central
ideas are a concept of termination-sensitive weak bisimilarity
tailored for constructive reasoning and the organization of 
evaluation in the delay-free semantics as an induction nested into coinduction.

Technically, we have carried out an advanced exercise in programming
and reasoning with mixed induction and coinduction, which we have also
formalized in Coq. The challenges in this exercise were both
mathematical and tool-related (Coq-specific). We deem that the
mathematical part was more interesting and important. The main new
aspect in comparison to our earlier development of coinductive
trace-based big-step semantics for noninteractive While was the need
to deal with definitions of predicates that nest induction into
coinduction---a relatively unexplored area in type theory.  In Coq, we
formalized them by parameterizing the inductive definition and
converting the coinductive definition into Mendler-like format.
Apparently, this technique is novel for the Coq community.

As future work, we would like to scale our development to concurrency.

\paragraph{Acknowledgments}

We thank Andreas Lochbihler, Nils Anders Danielsson and Thorsten
Altenkirch for discussions.

This research was supported by the Estonian Centre of Excellence in
Computer Science, EXCS, funded by the European Regional Development
Fund, and the Estonian Science Foundation grant no.~6940.


\appendix 

\section{Resumptions, Weak Bisimilarity, Delayful, Delay-Free and
  Classical-Style Big-step Semantics for While}

The notions of resumptions and weak bisimilarity and the evaluation
relations in the three big-step semantics shown of the main text are
fairly involved, because of the amount of detail. Therefore, we also
spell out what they specialize (or degenerate) to in the case of
ordinary non-interactive While, to better highlight the phenomena that
arise even in the absence of interaction.

\subsection{Resumptions, Bisimilarity, Weak Bisimilarity}

Delayful resumptions, with their strong bisimilarity, specialize to
delayed states $r : \res$ \`a la Capretta~\cite{Cap:genrct} defined
coinductively.
\[
\small
\infer={ \ret{\st} : \res}{ \st : \state}
\qquad
\infer={
  \delay{r} : \res
}{
  r : \res
}
\hspace*{3cm}
\infer={
  \bism{\ret{\st}}{\ret{\st}}
}{}
\qquad
\infer={
  \bism{\delay{r}}{\delay{r_*}}
}{
  \bism{r}{r_*}
}
\]
Convergence and (silent) divergence are defined inductively resp.\
coinductively; convergence reduces to termination at a final state.
\[
\small
\infer{
  \steps{\ret{\st}}{\ret{\st}}
}{}
\qquad
\infer{
  \steps{\delay{r}}{r'}
}{
  \steps{r}{r'}
}
\hspace*{3cm}
\infer={
  \uresp{\delay{r}}
}{
  \uresp{r}
}
\]
Responsiveness reduces to termination. Commitedness becomes
decidability between and termination and divergence. Commitedness is
tautologically true only classically.

Weak bisimilarity is defined in terms of convergence coinductively
exactly as Capretta~\cite{Cap:genrct} did.
\[
\small
\begin{array}{c}
\infer={
  \wkbism{r}{r_*}
}{
  \steps{r}{\ret{\st}} & \steps{r_*}{\ret{\st}}
}
\quad
\infer={
  \wkbism{\delay{r}}{\delay{r_*}}
}{
  \wkbism{r}{r_*}
}
\end{array}
\]
Any terminating delayed state can be normalized into a state. Any
decided delayed state can be normalized into a choice between a state
or a special divergence token.

\subsection{Delayful Semantics}

In the delayful big-step semantics, evaluation and extended evaluation
are defined mutually coinductively as follows.
\[
\small
\begin{array}{c}
\quad
\infer={\exec{\Assign{x}{e}}{\st}{\delay{(\ret{\update{\st}{x}{\eval{e}{\st}}})}}
}{}
\qquad
\infer={\exec{\Skip}{\st}{\ret{\st}}
}{}
\qquad
\infer={
  \exec{\Seq{s_0}{s_1}}{\st}{r'}
}{
  \exec{s_0}{\st}{r}
  &\execseq{s_1}{r}{r'}
}
\\[1ex]
\infer={
  \exec{\Ifthenelse{e}{s_t}{s_f}}{\st}{r}
}{
  \istrue{\st}{e}
  &\execseq{s_t}{\delay{(\ret{\st})}}{r}
}
\quad
\infer={
  \exec{\Ifthenelse{e}{s_t}{s_f}}{\st}{r}
}{
  \isfalse{\st}{e}
  &\execseq{s_f}{\delay{(\ret{\st})}}{r}
}
\\[1ex]
\infer={
  \exec{\While{e}{s_t}}{\st}{r'}
}{
  \istrue{\st}{e}
  &\execseq{s_t}{\delta{(\ret{\st})}}{r}
  &\execseq{\While{e}{s_t}}{r}{r'}
}
\quad
\infer={
  \exec{\While{e}{s_t}}{\st}{\delay{(\ret{\st})}}
}{
  \isfalse{\st}{e}
}
\\[3ex]
\infer={
  \execseq{s}{\ret{\st}}{r}
}{
  \exec{s}{\st}{r}
}
\qquad
\infer={
  \execseq{s}{\delay{r}}{\delay{r'}}
}{
  \execseq{s}{r}{r'}
}
\end{array}
\]
We have previously \cite{NU:trabco} conducted a thorough study of a
variation of this semantics (with intermediate states instead of
delays), explaining the design considerations in great detail. We have
also \cite{NU:hoalct} developed a Hoare logic for this semantics.

\subsection{Delay-Free Semantics}

Delay-free resumptions are the same as states.

In the delay-free semantics, there is one inductive evaluation
relation for terminating configurations. There is no need for a
separate extended evaluation relation (which would coincide with
evaluation anyhow, since resumptions and states are the same thing)
and no need to parameterize the evaluation relation.
\[
\small
\begin{array}{c}
\quad
\infer{\execcni{\Assign{x}{e}}{\st}{\update{\st}{x}{\eval{e}{\st}}}}{
}
\qquad
\infer{\execcni{\Skip}{\st}{\st}}{
}
\qquad 
\infer{
  \execcni{\Seq{s_0}{s_1}}{\st}{\st''}
}{
  \execcni{s_0}{\st}{\st'}
  &\execcni{s_1}{\st'}{\st''}
}
\\[1ex]
\infer{
  \execcni{\Ifthenelse{e}{s_t}{s_f}}{\st}{\st'}
}{
  \istrue{\st}{e}
  &\execcni{s_t}{\st}{\st'}
}
\quad
\infer{
  \execcni{\Ifthenelse{e}{s_t}{s_f}}{\st}{\st'}
}{
  \isfalse{\st}{e}
  &\execcni{s_f}{\st}{\st'}
}
\\[1ex]
\infer{
  \execcni{\While{e}{s_t}}{\st}{\st''}
}{
  \istrue{\st}{e}
  &\execcni{s_t}{\st}{\st'}
  &\execcni{\While{e}{s_t}}{\st'}{\st''}
}
\qquad
\infer{
  \execcni{\While{e}{s_t}}{\st}{\st}
}{
  \isfalse{\st}{e}
}
\end{array}
\]
The delay-free semantics agrees with the delayful semantics for
terminating delayed states.

It is the textbook big-step semantics of While, which accounts for
terminating configurations and assigns no evaluation result to
diverging configurations.

\subsection{Classical-Style Semantics}

A classical-style resumption is a state or the special token
$\bul$ for divergence.
\[
\small
\begin{array}{c}
\infer={ \retc{\st} : \resc}{ \st : \state}
\quad
\infer={\bul : \resc}{}
\end{array}
\]

The classical-style semantics has an inductively defined terminating
evaluation relation (defined exactly as that of the delay-free
semantics) and a coinductively defined diverging evaluation relation.
The latter depends on the former, but not the other way around. There
is no need for an extended evaluation relation.

\[
\small
\begin{array}{c}
\quad
\infer{\execcni{\Assign{x}{e}}{\st}{\update{\st}{x}{\eval{e}{\st}}}}{
}
\qquad
\infer{\execcni{\Skip}{\st}{\st}}{
}
\qquad 
\infer{
  \execcni{\Seq{s_0}{s_1}}{\st}{\st''}
}{
  \execcni{s_0}{\st}{\st'}
  &\execcni{s_1}{\st'}{\st''}
}
\\[1ex]
\infer{
  \execcni{\Ifthenelse{e}{s_t}{s_f}}{\st}{\st'}
}{
  \istrue{\st}{e}
  &\execcni{s_t}{\st}{\st'}
}
\quad
\infer{
  \execcni{\Ifthenelse{e}{s_t}{s_f}}{\st}{\st'}
}{
  \isfalse{\st}{e}
  &\execcni{s_f}{\st}{\st'}
}
\\[1ex]
\infer{
  \execcni{\While{e}{s_t}}{\st}{\st''}
}{
  \istrue{\st}{e}
  &\execcni{s_t}{\st}{\st'}
  &\execcni{\While{e}{s_t}}{\st'}{\st''}
}
\qquad
\infer{
  \execcni{\While{e}{s_t}}{\st}{\st}
}{
  \isfalse{\st}{e}
}
\\[2ex]
\infer={\execdni{s_0;s_1}{\st}}{
  \execdni{s_0}{\st}
}
\qquad
\infer={\execdni{s_0;s_1}{\st}}{
  \execcni{s_0}{\st}{\st'}
  &
  \execdni{s_1}{\st'}
}
\qquad
\infer={
  \execdni{\Ifthenelse{e}{s_t}{s_f}}{\st}
}{
  \istrue{\st}{e}
  &\execdni{s_t}{\st}
}
\quad
\infer={
  \execdni{\Ifthenelse{e}{s_t}{s_f}}{\st}
}{
  \isfalse{\st}{e}
  &\execdni{s_f}{\st}
}
\\[1ex]
\infer={
  \execdni{\While{e}{s_t}}{\st}}{
  \istrue{\st}{e}
  &\execdni{s_t}{\st}
}
\qquad
\infer={
  \execdni{\While{e}{s_t}}{\st}}{
  \istrue{\st}{e}
  &\execcni{s_t}{\st}{\st'}
  &\execdni{\While{e}{s_t}}{\st'}
}
\\[2.5ex]
\infer{\execC{s}{\st}{\st'}}{
  \execcni{s}{\st}{\retc{\st'}}
}
\quad
\infer{\execC{s}{\st}{\bul}}{
  \execdni{s}{\st}
}
\end{array}
\]

The classical-style semantics agrees with the delayful semantics for
decided delayed states (classically, any delayed state is decided).

A semantics in this spirit (with separate convergent and divergent
evaluation relations) was proposed for untyped lambda calculus by
Leroy and Grall \cite{LG:coibso}.

The delayful semantics (together with the identification of weakly
bisimilar delayed states) and the classical-style semantics have the
same purposes, but the delayful semantics is better behaved from the
constructive point-of-view. As a practical consequence, it has the
advantage that the evaluation relation can be turned into a function
(highly desirable, if one wants to be able to directly execute the
big-step semantics). This is not possible with the classical-style
semantics, as one would have to be able to decide whether a
configuration terminates before actually running it.


\end{document}